\documentclass[conference]{IEEEtran}
\IEEEoverridecommandlockouts

\usepackage{etoolbox}
\newtoggle{arxiv}
\toggletrue{arxiv} % arXiv-safe default; set \togglefalse{arxiv} for
\usepackage{cite}
\usepackage{amsmath,amssymb,amsthm,mathtools}
\usepackage[ruled,vlined,linesnumbered]{algorithm2e}
\usepackage{graphicx}
\usepackage{textcomp}
\usepackage{xcolor}
\usepackage[caption=false,font=footnotesize]{subfig}

\usepackage{enumitem}
\usepackage{bm}
\usepackage{booktabs}
\usepackage{dcolumn}
\usepackage[colorlinks=true,linkcolor=blue,citecolor=red,urlcolor=blue]{hyperref}
\usepackage{cleveref}

\crefname{appendix}{Appendix}{Appendices}
\Crefname{appendix}{Appendix}{Appendices}

\usepackage[acronym]{glossaries}
\glsdisablehyper

\newtheorem{proposition}{Proposition}

\newtheorem{remark}{Remark}

\newcommand{\bx}{\bm{x}}

\newcommand{\bz}{\bm{z}}
\newcommand{\bZ}{\bm{Z}}

\newcommand{\edges}{\mathcal{E}}

\newcommand{\subspace}{\mathcal{L}}

\newcommand{\spans}{\mathrm{span}}
\newcommand{\slacks}{s}
\newcommand{\slackprob}[1]{\tilde{\mu}_{\slacks_{#1}}}
\newcommand{\slackdens}[1]{\varphi_{\slacks_{#1}}}
\newcommand{\binprob}[1]{\mu_{\bZ}^{(#1)}}

\DeclareMathOperator*{\argmin}{arg\,min}

\newacronym{qpu}{QPU}{quantum processing unit}
\newacronym{cop}{CoP}{Coefficient of Performance}
\newacronym{sqa}{SQA}{simulated quantum annealer}
\newacronym{snis}{SNIS}{self-normalized importance sampling}
\newacronym{up}{UP}{unbalanced penalization}
\newacronym{method}{TAWP}{Topology-Aware Walsh--Fourier Penalization}
\newacronym{methodfull}{TAWP (full)}{Topology-Aware Walsh--Fourier Penalization}
\newacronym{methodtopo}{TAWP (topology)}{Topology-Aware Walsh--Fourier Penalization}
\newacronym{tts}{TTS}{time-to-solution}

\def\BibTeX{{\rm B\kern-.05em{\sc i\kern-.025em b}\kern-.08em
    T\kern-.1667em\lower.7ex\hbox{E}\kern-.125emX}}
\begin{document}

\title{
  Hardware-Aware QUBO Reformulation of Constrained Binary Optimization via the Walsh–Fourier Transform
}
\author{
\IEEEauthorblockN{
Loong Kuan Lee\textsuperscript{1},
Harsha Nagarajan\textsuperscript{2},\\
Thore Gerlach\textsuperscript{3},
Sascha M\"ucke\textsuperscript{4},
Ragavi Krishnamoorthy\textsuperscript{1},
Nico Piatkowski\textsuperscript{1}
}

\IEEEauthorblockA{
\textsuperscript{1}Fraunhofer IAIS, Germany, {\href{mailto:loong.kuan.lee@iais.fraunhofer.de}{loong.kuan.lee@iais.fraunhofer.de}}
}

\iftoggle{arxiv}{
    \IEEEauthorblockA{
      \textsuperscript{2}Applied Mathematics \& Plasma Physics (T-5), Los Alamos National Laboratory, USA, \href{mailto:harsha@lanl.gov}{harsha@lanl.gov}
    }
  }{
    \IEEEauthorblockA{
      \textsuperscript{2}Los Alamos National Laboratory, USA, \href{mailto:harsha@lanl.gov}{harsha@lanl.gov}
    }
  }

\iftoggle{arxiv}{
    \IEEEauthorblockA{
      \textsuperscript{3}European Space Agency (ESA), Advanced Concepts Team, Noordwijk, The Netherlands, \href{mailto:Thore.Gerlach@esa.int}{thore.gerlach@esa.int}
    }
  }{
    \IEEEauthorblockA{
      \textsuperscript{3}European Space Agency (ESA), The Netherlands, \href{mailto:Thore.Gerlach@esa.int}{thore.gerlach@esa.int}
    }
  }

\IEEEauthorblockA{
\textsuperscript{4}TU Dortmund University, Germany% \texttt{sascha.muecke@tu-dortmund.de}
}
}
% \author{\IEEEauthorblockN{1\textsuperscript{st} Loong Kuan Lee}
% \and
% \and

% NOTE: \makeatletter must wrap the \iftoggle, not sit inside its argument --
% etoolbox tokenizes the branch before \makeatletter would run, so an inner
% \def\@IEEEpubidpullup silently redefines \@ instead.
\makeatletter
\iftoggle{arxiv}{%
    % Keep the notice entirely in the bottom margin so the text block -- and
    % hence the page count -- is untouched: zero pullup, and \parbox[t] (whose
    % reference baseline is its FIRST line) grows the block downward from the
    % bottom of the text block rather than upward into the body.
    \def\@IEEEpubidpullup{0pt}%
    \IEEEpubid{\raisebox{-2.5\baselineskip}[0pt][0pt]{%
      \parbox[t]{\textwidth}{\footnotesize
      \textcopyright~2026 IEEE. To appear in \emph{IEEE QCE 2026}.
      Personal use of this material is permitted. Permission from IEEE must be
      obtained for all other uses, in any current or future media, including
      reprinting/republishing this material for advertising or promotional
      purposes, creating new collective works, for resale or redistribution to
      servers or lists, or reuse of any copyrighted component of this work in
      other works.\strut}}}%
    }{}
\makeatother
    
\maketitle

\begin{abstract}

  We present a novel slack-free, penalty-based framework for
  reformulating constrained binary optimization as Quadratic
  Unconstrained Binary Optimization (QUBO) on near-term quantum
  annealing hardware. Given a user-chosen penalty function that most
  naturally captures a constraint---typically non-quadratic, such as a
  Heaviside-function surrogate---and a target probability measure over
  the Boolean hypercube, our method returns the weighted least-squares
  projection of the chosen penalty function onto the subspace spanned by
  linear and quadratic Walsh--Fourier characters that correspond to
  physically realizable couplings on the target hardware graph. Within
  this restricted family, the resulting quadratic surrogate is
  \emph{uniquely} and \emph{optimally} determined by the normal
  equations: unlike state-of-the-art approaches, it introduces no
  per-constraint penalty coefficients to tune and avoids dense all-pairs
  couplings by construction. Two practical consequences follow. First,
  the projected penalty respects device connectivity, reducing chain
  lengths and physical-qubit overhead after minor embedding. Second, we
  show empirically that this hardware-native surrogate can outperform
  denser full-pairwise projections, despite being drawn from a strictly
  smaller approximation space. This advantage widens once the QUBO is
  embedded and sampled on quantum annealers, yielding samples with the
  lowest worst-case and mean objective gaps compared to unbalanced
  penalization and a hardware-blind projection onto all quadratic
  terms.
  
  % Therefore, our approach sits in between penalty methods---such as
  % unbalanced penalization---that requires quadratic terms between each
  % logical qubit; and approaches that restrict themselves to just
  % linear penalty functions.

  % with inequality constraints on quantum annealers without
  % introducing additional slack variables. The core idea is to start from
  % an ``ideal'' (generally non-quadratic) penalty for each violated
  % inequality, expand that penalty in the Walsh--Hadamard/Fourier basis
  % over $\{-1,1\}^n$, and then project it onto the lower-dimensional
  % function space that is physically realizable on a target annealer
  % graph. This yields a connectivity-aware quadratic penalty with a
  % closed-form expression for its optimal coefficients under mean-squared
  % error. We provide the projection theorem and proof, practical
  % constructions for step, sigmoid, ReLU, hinge, and cubic penalties, and
  % an implementation outline suitable for readers with a background in
  % quantum optimization.
\end{abstract}


\section{Introduction}
\label{sec:intro}

Constrained binary optimization underlies a wide range of industrial
decision problems, including unit commitment in power
systems~\cite{barrass2025leveraging,paterakis2023hybrid}, portfolio
construction under budget and cardinality
constraints~\cite{phillipson2021portfolio}, and vehicle routing, traffic
flow, and multi-agent
pathfinding~\cite{neukart2017traffic,gerlach2025hybrid}. Lucas's survey
catalogs Ising formulations for all of Karp's 21 NP-complete
problems~\cite{lucas2014ising}, and subsequent work has continued to
expand this collection~\cite{glover2022quantum}. A common template is
the binary linear program (BLP) with inequality
constraints~\cite{wolsey2020integer,woeginger1992subset},
\begin{equation}\label{eq:blp}
  \min_{\bx\in\{0,1\}^n}\; c^\top\bx
  \quad\text{s.t.}\quad A\bx\geqslant b,
\end{equation}
together with quadratic-objective generalizations. In all such settings,
the inequalities encode hard feasibility constraints rather than soft requirements.

Quantum annealing is closely related to adiabatic quantum
computing, in which a system is evolved from an easily prepared
initial Hamiltonian toward a problem Hamiltonian whose low-energy
states encode candidate solutions~\cite{kadowaki1998quantum}.
Commercial quantum annealers realize this paradigm for programmable
Ising models, while gate-model approaches such as the quantum
approximate optimization algorithm (QAOA)~\cite{farhi2014qaoa,hadfield2019quantum}
provide a digital counterpart. On current NISQ devices, however,
performance is shaped by noise, limited control precision, and sparse
hardware connectivity. This challenge is especially acute for hard
inequality constraints: while objectives are often readily expressed
in binary form, enforcing feasibility within a sparse two-body
Hamiltonian is not. Consequently, the QUBO formulation is central, as
dense couplings and excessive auxiliary variables can increase embedding
overhead and degrade solution quality~\cite{coffrin2019evaluating}.

The standard starting point is the slack-variable construction: a
nonnegative slack \(s_j \geqslant 0\) rewrites
\(a_j^\top\bx \geqslant b_j\) as \(a_j^\top\bx - s_j = b_j\), enforced
by the quadratic penalty
\(\lambda_j(a_j^\top\bx - s_j - b_j)^2\)~\cite{glover2022quantum}.
This is a penalty relaxation rather than an exact reformulation. Only
a sufficiently large \(\lambda_j\) guarantees that the ground state
recovers the feasible optimum. On near-term hardware it incurs three
familiar costs: slack variables must be binarized, adding auxiliary
qubits per constraint~\cite{montanez-barrera2024}; the squared-slack
term induces a dense quadratic block that requires long ferromagnetic
chains under embedding, increasing qubit usage and chain-break
risk~\cite{zbinden2020embedding}; and the added slack bits enlarge the
search space and can introduce spurious suboptimal
minima~\cite{montanez2023improving}. Sparsifying the resulting QUBO
may further increase the lifted dimension~\cite{suda2026}.

These drawbacks motivate several alternative lines of work. At the
circuit level, QAOA-based approaches encode constraints natively via
feasible-subspace mixers~\cite{bucher2025penalty}, tailored driver
Hamiltonians~\cite{bucher2025efficient}, or problem-specific state
preparation~\cite{christiansen2025quantum}, but they belong to a
different algorithmic regime and inherit the depth and control
limitations of gate-model implementations.

For annealing and hybrid workflows, a prominent alternative eliminates
slack variables through iterative Lagrangian updates---including
ADMM-type schemes~\cite{yonaga2020solving,mucke2023efficient},
augmented-Lagrangian methods~\cite{djidjev2023quantum,sharma2025}, and
subgradient procedures~\cite{takabayashi2025a}---at the cost of outer
iterations that repeatedly retune multipliers and re-solve the
underlying QUBO. Decomposition avoids monolithic QUBOs altogether,
handling feasibility in a column-generation master
problem~\cite{gerlach2025hybrid} or generating improving variables via
QUBO-based pricing subproblems~\cite{takabayashi2025}, but shifts the
difficulty to coordinating repeated quantum subroutine calls.

Closest in spirit to our setting are slack-free penalty methods, which
replace exact encodings by directly designed quadratic
surrogates---most prominently \emph{unbalanced
penalization}~\cite{montanez-barrera2024,lee2025a,montanez2023improving},
which penalizes constraint violation more strongly than it rewards
satisfaction via per-constraint parameters
\((\lambda_{1,j},\lambda_{2,j})\), without auxiliary variables.
However, two concerns remain: the per-constraint parameters lack a
closed-form characterization and are typically tuned heuristically,
e.g.\ by Nelder--Mead, and quadratic terms still couple all
variables in a constraint's support, so the resulting QUBO remains
dense and can embed poorly on sparse hardware graphs.

These limitations motivate the central problem addressed in this paper:
\textit{For a given constraint penalty, how can one construct a
  quadratic surrogate that best approximates that penalty while using
  only couplings natively supported by target hardware?} A
solution to this problem would simultaneously eliminate per-constraint
penalty tuning and avoid the embedding overhead associated with dense
penalty constructions.

Our Walsh--Fourier-based construction also belongs to the slack-free
family, but differs in how the quadratic penalty is obtained. Rather
than positing a manually tuned quadratic ansatz, we compute the
orthogonal projection of a user-chosen penalty function \(\psi\) in a
weighted \(L_2\) space over the Boolean cube onto the: constant,
one-body, and admissible two-body terms. If all pairwise terms are
admissible, this gives the best quadratic surrogate in the chosen
weighted least-squares sense; if the admissible two-body terms are
restricted to an edge set \(\edges\), it gives the best
hardware-admissible surrogate for the corresponding hardware graph.

Our \textit{main contribution} is a projection-based framework for constructing
\acrfull{method}: hardware-aware quadratic surrogates derived
from a user-chosen penalty by a single weighted least-squares
projection. In particular, the framework provides:
\begin{enumerate}
  \item \textbf{Penalty fit:} a principled quadratic approximation of a
  chosen, possibly non-quadratic, penalty function; and
  \item \textbf{Topology fit:} the corresponding optimal approximation
  when the admissible two-body terms are restricted to a prescribed
  coupling set \(\edges\), such as the native couplers induced by a
  placement on a target hardware graph.
\end{enumerate}
In summary, the framework decouples the modeling choice of penalty from
the hardware-imposed choice of topology: the user specifies what to
penalize, the device specifies which couplings are admissible, and a
single weighted least-squares system produces the best surrogate
consistent with both. When \(\edges\) is induced by a native placement,
the resulting QUBO is chain-free by construction.
\Cref{tab:sharma-lau-constraint-encoding} summarizes how this approach
compares with the main alternative paradigms.

Furthermore, in \Cref{sec:benchmark} we evaluate our framework on real
quantum hardware using a standard set of multidimensional knapsack
benchmark instances. For most instances we observe two consistent
benefits: (1) the fully quadratic approximation produces the largest
number of feasible samples from the \acrfull{qpu} and yields the
solution with the smallest objective gap; and (2) the topology-aware
projection has the lowest maximum and mean objective gaps over its
samples from the quantum annealer. Together, these results indicate that
our projection-based framework is a promising approach for encoding
constraint penalties when formulating QUBO models.

Note that although we present and evaluate the framework for inequality
constraints---the case in which a slack-free QUBO reformulation is most
challenging---the projection itself applies to any pseudo-Boolean
function; \Cref{rem:scope} makes the scope precise, covering equality
constraints and, more generally, the direct projection of (higher-order)
objectives.

The next section presents the necessary mathematical background, followed
by the main results in \Cref{sec:method}.

\begin{table}[t]
  \caption{Comparison of constraint encoding methods~\cite{sharma2025}.}
  \label{tab:sharma-lau-constraint-encoding}
  \centering
  \resizebox{\columnwidth}{!}{%
  \begin{tabular}{@{}lllll@{}}
    \toprule
    Method & Qubit Overhead & Feasibility Guarantee & Computational Cost & Scalability \\
    \midrule
    Slack Variables & High & Guaranteed & Moderate & Low \\
    Penalty Methods & None & Not guaranteed & Low (single shot) & Moderate \\
    Lagrangian Relaxation & None & Strong (iterative) & High (multiple runs) & High \\
    Proposed Method & None & Not guaranteed & Low (single shot) & High \\
    \bottomrule
  \end{tabular}%
  }
\end{table}

\section{Quantum Annealing and Mathematical Preliminaries}\label{sec:preliminaries}
\subsection{Quantum Annealing and Hardware Constraints}
\label{sec:back-qa}
In standard quantum annealing, the system Hamiltonian interpolates
between a transverse-field driver and a programmable Ising problem
Hamiltonian, \(H(s)=A(s)\,H_{\mathrm D}+B(s)\,H_{\mathrm P}\) for
\(s\in[0,1]\), where \(H_{\mathrm D}=-\sum_i \sigma_i^x\) and
\(H_{\mathrm P}=\sum_i h_i \sigma_i^z +\sum_{(i,j)\in\edges_{\mathrm H}}
J_{ij}\,\sigma_i^z\sigma_j^z\) is the Ising form of a QUBO under the
usual binary-to-spin transformation. In the ideal closed-system
adiabatic limit, and under suitable gap conditions, sufficiently slow
evolution keeps the state near the instantaneous ground state; on
present-day open-system, finite-temperature annealers, the device
instead returns samples biased toward low-energy configurations of
\(H_{\mathrm P}\) rather than a guaranteed ground state.

The relevant hardware constraint is the sparsity of the user-accessible
coupler graph \(G_{\mathrm H}=(V_{\mathrm H},\edges_{\mathrm H})\).
D-Wave's 2000Q, Advantage, and Advantage2 systems use the Chimera,
Pegasus, and Zephyr topologies, respectively, with successively higher
per-qubit coupler degree~\cite{boothby2020pegasus}. Any logical
Ising/QUBO interaction graph not contained in \(G_{\mathrm H}\) must be
minor-embedded~\cite{cai2014practical}, representing each logical
variable by a ferromagnetically coupled chain of physical qubits.  This
increases physical-qubit overhead, forces a chain-strength trade-off
between chain integrity and distortion of the logical problem, and
introduces broken-chain events that must be resolved during unembedding.

Throughout the paper, we write
\(\edges=\bigl\{\{i,j\}\subseteq V_{\mathrm L}:
(\pi(i),\pi(j))\in\edges_{\mathrm H}\bigr\}\) for the set of
logical couplings available under a chosen injective placement
\(\pi: V_{\mathrm L}\hookrightarrow V_{\mathrm H}\). Any logical
QUBO whose interaction graph is contained in \(\edges\) admits a
native, chain-free placement under \(\pi\), avoiding the qubit
overhead, coupling-range trade-offs, and chain-break postprocessing
of minor embedding.

\subsection{Pseudo-Boolean Functions and Walsh--Fourier Transform}
A pseudo-Boolean function is a real-valued function on the binary
hypercube, \(\psi : \{0,1\}^n \to \mathbb{R}\). We equivalently view the
same function on the spin cube \(\{-1,1\}^n\) via the standard
binary-to-spin bijection, \(\bz_i = 2\bx_i - 1\) and
\(\bx_i = \tfrac{1}{2}(1 + \bz_i)\). This bijection preserves
minimizers under the corresponding change of variables and, for
quadratic functions, yields the usual equivalence between QUBO
coefficients in \(\bx\) and Ising coefficients in \(\bz\), up to
additive constants. Under this equivalence, both the objective and
penalty functions considered in this paper are pseudo-Boolean.

%\subsection{The Walsh--Fourier Transform}
Every pseudo-Boolean function
\(\psi:\{-1,1\}^n \to\mathbb{R}\) admits a unique multilinear expansion
\begin{equation}\label{eq:fourier}
  \psi(\bz) \;=\; \sum_{S\subseteq[n]} \bm{\theta}_S \prod_{i\in S} \bz_i
  \;=\; \sum_{S\subseteq[n]} \bm{\theta}_S\,\chi_S(\bz),
\end{equation}
with real coefficients \(\bm{\theta}_S\in\mathbb{R}\), where
\(\chi_S(\bz)=\prod_{i\in S} z_i\) are the
\emph{Walsh characters}~\cite{walsh1923}. Given a probability measure
\(\mu\) on \(\{-1,1\}^n\), define
\begin{equation}\label{eq:inner-prod}
  \langle \psi, \phi \rangle_\mu
  \;\coloneq\; \mathbb{E}_{\bz\sim\mu}
  \bigl[\psi(\bz)\,\phi(\bz)\bigr].
\end{equation}
When \(\mu\) has full support, this defines an inner product on the
finite-dimensional space of real-valued functions on the cube, which we
denote by \(L_2(\mu)\). Under the uniform measure
\(\mu_{\mathrm U}(\bz)=2^{-n}\), the Walsh characters form an
orthonormal basis of
\(L_2(\mu_{\mathrm U})\)~\cite{morgenthaler1957,odonnell2014}. In this
case, the expansion in~\eqref{eq:fourier} is an orthonormal change of
basis, and the coefficients are given by
$\bm{\theta}_S = \langle \psi, \chi_S \rangle_{\mu_{\mathrm U}}$.

For a family of subsets \(\Phi\subseteq 2^{[n]}\), let
\(
  \subspace
  = \spans\{\chi_S:S\in\Phi\}
  \subseteq L_2(\mu_{\mathrm U}).
\)
The orthogonal projection of \(\psi\) onto \(\subspace\) is obtained
by truncating the Walsh--Fourier expansion to the indices in \(\Phi\),
\begin{equation}\label{eq:ortho-proj}
  \hat\psi_\subspace(\bz)
  \;=\; \sum_{S\in\Phi} \bm{\theta}_S\,\chi_S(\bz).
\end{equation}
By the Hilbert projection theorem~\cite{rudin1991,kreyszig1978functional},
this projection is the best approximation to \(\psi\) in \(\subspace\),
\begin{equation}\label{eq:ortho-proj-opt}
  \hat\psi_\subspace
  =
  \argmin_{\phi\in\subspace}
  \|\psi-\phi\|^2_{\mu_{\mathrm U}}.
\end{equation}
The truncation in~\eqref{eq:ortho-proj} relies on orthonormality
under the uniform measure.

\subsection{Limitations of the Uniform Measure}\label{sec:uniform-bad}
\begin{figure}
  \centering
  \includegraphics[width=1.0\linewidth]{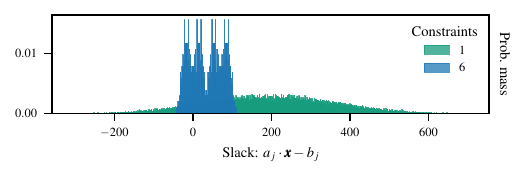}
  \caption{Distribution of slack values
    \(s_j(\bx)=a_{j}^\top\bx - b_{j}\) under uniform measure on
    \(\{0,1\}^n\), for two constraints of the SAC-94 \texttt{pet3}
    MDKP instance~\cite{drake2015}.}
  \vspace{-0.2cm}
  \label{fig:slacks}
\end{figure}

Generally, for any penalty \(\psi\) that depends on \(\bx\) only through
the slack \(s_j(\bx)=a_j^\top\bx-b_j\) and any full-support measure
\(\mu\), the squared error of a surrogate \(\phi\) splits across the
attainable slack values, writing $\bm{s}_{j}\coloneq s_{j}(\bx)$ via abuse
of notation,
\begin{gather*}
  \|\psi-\phi\|^2_{\mu}
  = \sum_{s}
  \mathbb{P}_{\mu}\bigl[\bm{s}_j=s\bigr]
  \mathbb{E}_{\mu}\bigl[\bigl(\psi(\bx)-\phi(\bx)\bigr)^2
  \,\big|\, \bm{s}_j=s\bigr].
\end{gather*}
Choosing \(\mu\) is thus essentially choosing an \emph{error budget},
where high-mass slack values are fit accurately and vice versa.
However, the uniform measure offers no such choice as its weights
\(\mathbb{P}_{\mu_{\mathrm U}}[\bm{s}_j=s]\) are the instance's slack
histogram, which concentrates wherever the constraint dictates and
varies widely between constraints; see \Cref{fig:slacks}.

These limitations motivate replacing the uniform measure with a
full-support alternative $\mu$ that concentrates its mass on the slack
regions that matter; \Cref{sec:method-config} makes this concrete with a
Gaussian target centered at mid-feasible slack. However, under $\mu$, the
Walsh characters are generally not orthogonal in
\(L_2(\mu)\). Therefore, the projection coefficients no longer equal the
truncated uniform Fourier coefficients and must instead be obtained from
the weighted least-squares normal equations which we will develop next.

\begin{figure*}[t]
  \centering
  \includegraphics[width=\textwidth]{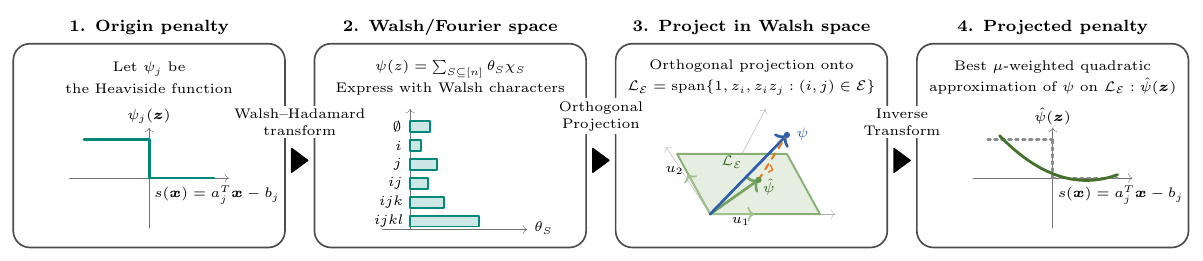}
  \caption{Connectivity-aware penalty approximation: origin penalty
    $\psi$ is expressed in the Walsh basis, orthogonally
    projected onto the hardware-admissible subspace
    $\mathcal{L}_{\edges}=\mathrm{span}\{1,z_i,z_i
    z_j:\{i,j\}\in\edges\}$, and synthesized back as
    $\hat{\psi}_{\edges}$---the best $\mu$-weighted approximation of
    $\psi$ on the couplings allowed by $\edges$.}
  \label{fig:gram-schmidt-projection}
\end{figure*}

\section{Topology-Aware Walsh--Fourier Penalization}
\label{sec:method}
Throughout this section we are interested in a constrained binary
optimization problem with the general form,
\begin{equation}
  \label{eq:opt-prob}
  \min_{\bx\in\{0,1\}^n}\; f(\bx)
  \quad\text{s.t.}\quad A\bx\geqslant b.
\end{equation}
Passing to spin variables via $\bz_i=2\bx_i-1$, we represent the penalty in
the spin domain as $\psi : \{-1,1\}^n \to\mathbb{R}$ and write the
penalized problem as
\begin{equation}
  \label{eq:min-prob}
  \min_{\bz\in\{-1,1\}^n}\; f(\bz) + \lambda \sum_{j=1}^{m}\psi_{j}(\bz; A, b),
\end{equation}
where the multiplier $\lambda$ controls the strength of the penalty
term.

\subsection{Summary of Main Result}\label{sec:main-res}
Concretely, our framework requires three inputs: a user-chosen
``origin'' penalty function $\psi:\{-1,1\}^{n}\to\mathbb{R}$ for one
inequality constraint, a probability measure $\mu_{\bZ}$ with full
support over $\bz\in\{-1,1\}^{n}$---notated as $\mu$ for brevity, and a
set $\edges$ of admissible logical couplings induced by the target
\acrshort{qpu} placement. Define
\begin{gather}
  \Phi \coloneq \{\emptyset\}
  \cup \bigl\{\{i\}: i\in[n]\bigr\}
  \cup \edges,\\
  \label{eq:subspace}
  \subspace_{\edges}
  \coloneq \spans\{\chi_S : S\in\Phi\} \subseteq L_2(\mu).
\end{gather}
Thus $\subspace_{\edges}$ contains the constant term, all linear Walsh
characters, and only those quadratic Walsh characters whose couplings
are allowed by \(\edges\).

\Cref{fig:gram-schmidt-projection} then summarizes the core idea of our
framework geometrically. The origin penalty $\psi$ is first expressed in
terms of the Walsh characters $\{\chi_S : S \subseteq [n]\}$. Then
projecting $\psi$ onto the new subspace $\subspace_{\edges}$, under the
$\mu$-weighted inner product, results in $\hat{\psi}_{\edges}$, the best
approximation of $\psi$ using only Walsh characters supported by the
target topology. The coefficients of this projected penalty can be read
directly as the corresponding QUBO/Ising penalty terms,
\begin{equation}
  \label{eq:pen-approx}
  \hat{\psi}_{\edges}(\bz)
  \coloneq \sum_{S\in\Phi} \bm{\hat{\theta}}_{S}\chi_{S}(\bz).
\end{equation}
Writing $\Phi=\{S_1,\ldots,S_K\}$, the Walsh-coordinate vector
$\bm{\hat{\theta}}\in\mathbb{R}^K$ satisfies the normal equations
\begin{equation}
  \label{eq:normal}
  G\bm{\hat{\theta}}=\bm{c},
\end{equation}
where
$G_{i,j}= \langle \chi_{S_i},\chi_{S_j}\rangle_{\mu}$ and
$\bm{c}_{i}= \langle \psi,\chi_{S_i}\rangle_{\mu}$.

In practice, the expectations defining \(G\) and \(\bm{c}\)
are estimated by \acrfull{snis}, after which we solve the regularized
system,
\[
  (\hat{G}+\epsilon I)\bm{\hat{\theta}}=\bm{\hat{c}},
\]
where $\epsilon$ is a small regularization parameter. An algorithmic
summary of this practical procedure is given in \Cref{alg:wls}.

\Cref{subsec:sampled-projection} gives the sampling-based
implementation, \Cref{subsec:target-measure} explains how we construct
the nonuniform target measure, and \Cref{sec:method-config} specifies
the fixed projection design used in the experiments below.

\begin{remark}[Scope of applicability]\label{rem:scope}
The construction above is not limited to inequality constraints. The projection applies to an arbitrary pseudo-Boolean function $\psi$ and returns its best approximation in $\subspace_{\edges}$. Equality constraints $a_j^\top\bx=b_j$ can be handled using the two-sided penalty $(a_j^\top\bx-b_j)^2$, together with a target measure concentrated near $s_j=0$ while retaining full support, as discussed in \Cref{sec:uniform-bad}, particularly when $\edges$ is sparse.
  % They are, however, the easier case. The penalty is already
  % exactly quadratic, so only the topology fit is nontrivial, whereas
  % inequalities require both the penalty fit and the topology fit.
  % The normal equations in \eqref{eq:normal} could likewise project the
  % objective itself---sparsifying hardware-incompatible couplings or
  % reducing higher-order objectives to quadratic ones. We restrict our
  % claims to constraint penalties nonetheless, since sparsifying an
  % objective function could change the true minimizer. We therefore
  % leave minimizer-preserving objective projection to future work.
\end{remark}

\subsection{Exact Projection in Walsh--Fourier Coordinates}\label{subsec:exact-projection}
In order to show why the coefficients $\bm{\hat{\theta}}$ satisfy the
normal equations as defined in \eqref{eq:normal}, we first characterize
the best approximation of $\psi$ in the subspace $\subspace_{\edges}$ as
the surrogate penalty function that minimizes the squared distance to
$\psi$:
% Fortunately, Lax et al. \cite{lax2005} showed that for a non-uniform
% measure $\mu$\footnote{We assume $\mu$ has the same support as a uniform
% distribution.}, the best approximation of a function
% $\psi\in\fullspace(\mu)$ in the subspace
% $\subspace\subseteq\fullspace(\mu)$ is just the unique function,
% \begin{equation}
%   \hat{\psi}_{\subspace} = \sum_{j=1}^{M}\inner{\psi, u_{j}}_{\mu}u_{j}
%   \label{eq:pen-proj}
% \end{equation}
% where $\bm{u}\coloneq u_{1},\ldots,u_{M}$ is an orthonormal basis for
% $\subspace$ obtained by applying Gram--Schmidt on the Walsh characters
% that forms a basis for $\subspace$,
% \begin{equation}
%   \mathcal{L}=
%   \mathrm{span}\{\chi_{S_1},\ldots,\chi_{S_K}\}\subseteq L_2(\mu).
% \end{equation}

% However, recall from \Cref{sec:contrib} that that our main goal is to
% obtain an approximation of $\psi$ as a linear combination of Walsh
% characters,
% \begin{gather}
%   \hat{\psi}_{\subspace}(\bz) =
%   \sum_{i=1}^{K} \hat{\theta}_{i} \chi_{\Phi_{i}}(\bz),
%   \label{eq:pen-approx}
%   \intertext{where}
%   K \coloneq |\Phi| \qquad
%   \Phi \coloneq \{\emptyset\} \cup [n] \cup \edges
% \end{gather}
% so that we can construct a QUBO penalty matrix directly from
% $\{\hat{\theta}_{i}\}_{i=1}^{K}$. Therefore the approximation in
% \Cref{eq:pen-proj} with respect to the orthonormal basis under $\mu$ is
% insufficient.

% Therefore instead of using Gram--Schmidt, we will use a weighted
% least-squares approach. Recall that $\hat{\psi}_{\subspace}\in\subspace$
% is the ``best approximation'' of $\psi$ if and only if,
\begin{equation}
  \label{eq:best-approx}
  \hat{\psi}_{\edges}
  = \argmin_{\phi\in\subspace_{\edges}} \|\psi - \phi\|_{\mu}^{2}.
\end{equation}
Writing $\Phi=\{S_1,\ldots,S_K\}$ and
\(\phi(\bz)=\sum_{i=1}^{K}\bm{\hat{\theta}}_{i}\chi_{S_i}(\bz)\), we have
\begin{align*}
  J(\bm{\hat{\theta}})
  &= \mathbb{E}_{\bz\sim\mu}\bigl[(\psi - \phi)^{2} \bigr]
  =\|\psi\|_{\mu}^{2} - 2\bm{\hat{\theta}}^{\top}\bm{c}
    + \bm{\hat{\theta}}^{\top} G \bm{\hat{\theta}},
\end{align*}
where $G_{ij}=\langle \chi_{S_i},\chi_{S_j}\rangle_{\mu}$ and
$\bm{c}_i=\langle \psi,\chi_{S_i}\rangle_{\mu}$.
Taking the gradient of $J(\bm{\hat{\theta}})$, we obtain
\begin{align*}
  \nabla_{\bm{\hat{\theta}}} J(\bm{\hat{\theta}})
  =-2\bm{c} + (G + G^{\top})\bm{\hat{\theta}}
  =2G\bm{\hat{\theta}} - 2\bm{c},
\end{align*}
since $G = G^{\top}$ is symmetric. Setting this gradient to zero yields
the ``normal equations'' in \eqref{eq:normal}, and hence any coefficient
vector $\bm{\hat{\theta}}$ that satisfies \eqref{eq:normal} is a
stationary point of $J$. Moreover, $G$ is the Gram matrix of the family
$\{\chi_{S_i}\}_{i=1}^{K}$, so $G\succeq0$. Because $\mu$ has full
support on $\{-1,1\}^{n}$ and the Walsh characters are linearly
independent, this Gram matrix is in fact positive definite,
$G\succ0$. Therefore the Hessian $2G$ is positive definite, $J$ is
strictly convex, and the solution of \eqref{eq:normal} is the unique
global minimizer~\cite{kreyszig1978functional,boyd2004convex}.

The coefficients $\{\bm{\hat{\theta}}_{i}\}_{i=1}^{K}$ in
\eqref{eq:pen-approx} can therefore be obtained simply by solving the
linear system in \eqref{eq:normal}. In particular, if
\(\psi\in\subspace_{\edges}\), the unique minimizer of
\eqref{eq:best-approx} is \(\psi\) itself; in other words, the exact
projection is lossless on penalties already representable within the
admissible coupling set.

\subsection{Practical Sampling-Based Implementation}\label{subsec:sampled-projection}
\begin{algorithm}[t]
  \caption{\acrshort{method} via weighted least-squares}\label{alg:wls}
% \KwIn{Origin penalty in spin variables--$\psi$;
%   Target probability measure--$\mu_{\bZ}$;
%   Proposal probability measure--$\nu_{\bZ}$;
%   Sample size--$N$;
%   Index tuples of coupled variables--$\edges$;
%   Small regularization term--$\lambda$
% }
\KwIn{$\psi$, $\mu_{\bZ}$, $\nu$, $N$, $\edges$, $\epsilon$}
\KwOut{$\hat{\psi}_{\edges}$}
\KwData{Number of variables--$n$}
$\Phi \gets \{\emptyset\}
\cup \bigl\{\{i\}: i\in[n]\bigr\}
\cup\edges$\;
% $\bm{g}_{\Phi}(\bz)=
% \bigl[\prod_{i\in S}\bz_{i} : S \in \Phi\bigr]$\;
% $\hat{G} \gets $ zero $|\Phi|\times |\Phi|$ matrix\;
% $\bm{\hat{c}} \gets $ zero $|\Phi|$-length vector\;
% $W\gets 0$\;
\For{$k \in [N]$}{
  $\bz^{(k)}\sim\nu$\;
  $w_{k}\gets \frac{\mu_{\bZ}(\bz^{(k)})}{\nu(\bz^{(k)})}$\;
  % $W\gets W+w_k$\;
  % $\hat{G} \gets \hat{G} + w_{k}
  % \bm{g}_{\Phi}(\bz^{(k)})\bm{g}_{\Phi}(\bz^{(k)})^{\top}$\;
  % $\bm{\hat{c}} \gets \bm{\hat{c}} + w_{k}
  % \psi(\bz^{(k)})\bm{g}_{\Phi}(\bz^{(k)})$\;
}
% $\hat{G}\gets \hat{G}/W$\;
% $\bm{\hat{c}}\gets \bm{\hat{c}}/W$\;
$\hat{G},\bm{\hat{c}} \gets $ estimates in \eqref{eq:est-G} and
\eqref{eq:est-c}, respectively\;
$\bm{\hat{\theta}} \gets $ solution to linear system $(\hat{G} + \epsilon
I)\bm{\hat{\theta}} = \bm{\hat{c}}$\;
\Return{$\hat{\psi}_{\edges}(\bz)
  \coloneq\sum_{S\in\Phi} \bm{\hat{\theta}}_{S}\prod_{i\in S}\bz_{i}$}
\end{algorithm}

However, there are two main barriers to a practical implementation of
this approach:
\begin{enumerate}
\item since the size of the domain $\{-1,1\}^{n}$ is exponential in $n$, we can
  realistically only estimate $G$ and $\bm{c}$ from samples of
  $\bz\sim\mu$, and
\item in order to avoid the issue described in \Cref{sec:uniform-bad},
  we have to define a probability measure over the slack domain first,
  namely $\slackprob{j}$, and then construct the corresponding
  spin-domain measure $\binprob{j}$.
\end{enumerate}
To address both issues, we draw samples
$\bz^{(k)}\sim \nu$, $k=1,2,\ldots,N$, from a simple proposal
distribution $\nu$. Then, using \acrshort{snis}, we estimate $G$ and
$\bm{c}$ by
\begin{gather}
  \label{eq:est-G}
  \hat{G} \coloneq \frac{1}{\sum_{k}w_{k}}
  \sum_{k=1}^{N} w_{k} \bm{g}_{\Phi}(\bz^{(k)})
  \bm{g}_{\Phi}(\bz^{(k)})^{\top},\\
  \label{eq:est-c}
  \bm{\hat{c}} \coloneq \frac{1}{\sum_{k}w_{k}}
  \sum_{k=1}^{N} w_{k} \psi(\bz^{(k)}) \bm{g}_{\Phi}(\bz^{(k)}),
\end{gather}
where $w_{k} \coloneq \mu_{\bZ}(\bz^{(k)})/\nu(\bz^{(k)})$ and
$\bm{g}_{\Phi}(\bz) = \bigl[\chi_{S}(\bz) : S \in \Phi \bigr]^{\top}$.
The coefficients in \eqref{eq:pen-approx} are then obtained by solving
the linear system
\begin{equation}
  (\hat G+\epsilon I)\bm{\hat{\theta}}=\bm{\hat{c}},
\end{equation}
with optional small $\epsilon\geqslant0$ for numerical
stability. \Cref{alg:wls} summarizes this entire procedure. The
worst-case time complexity of this approach is then
$\mathcal{O}(N |\Phi|^{2} + |\Phi|^{3})$. A more detailed description
and proof can be found in \Cref{prop:complexity} in \Cref{apx:proofs}. Note
that since the topologies of D-Wave \acrshortpl{qpu} have constant
degree, the number of edges, and therefore $|\Phi|$, scales linearly with
the number of variables.
The remaining question is how to compute $\binprob{j}(\bz)$ when we
only know $\slackprob{j}$; equivalently, how do we transform
$\slackprob{j}$ into $\binprob{j}$?

\subsection{Transforming a Slack-Domain Measure to a Spin-Domain Measure via the Lugannani--Rice Formula}\label{subsec:target-measure}
\begin{figure}
  \centering
  \includegraphics[width=1.0\linewidth]{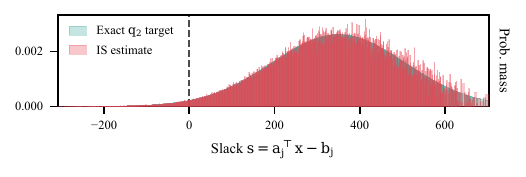}
  \caption{Exact target bin masses $u_l$ vs. their
    importance-sampling estimates for constraint $j=1$ of \texttt{pet3};
    estimates track the exact masses closely.}
  \label{fig:slack-dist}
\end{figure}

To construct this measure, assume that the discrete probability measure
$\slackprob{j}$ is obtained by discretizing a density function
$\slackdens{j}$ on the slack domain. For instance, $\slackdens{j}$
could be the Gaussian target defined in \eqref{eq:gaussian-target}.
Partition the slack axis into bins $B_l=[\ell_l^-,\ell_l^+)$ of width
$\Delta$, define the spin-domain slack by
\(s_j(\bz)\coloneq a_j^\top(1+\bz)/2-b_j\), and let $l(\bz)$ be the
index such that $s_j(\bz)\in B_{l(\bz)}$.
We then define two probability mass functions on the bins:
\begin{enumerate}
\item the target bin masses $u_l \propto \int_{B_l} \slackdens{j}(s)\,ds$, and
\item the proposal bin masses $v_l \coloneq \mathbb{P}_{\nu}\bigl[s_j(\bz)\in B_l\bigr]$.
\end{enumerate}
The induced spin-domain target measure is then
\begin{gather}
  \label{eq:bin-prob}
  \binprob{j}(\bz)\propto \nu(\bz)\,\frac{u_{l(\bz)}}{v_{l(\bz)}}, \qquad
  w_k \propto \frac{u_{l(\bz^{(k)})}}{v_{l(\bz^{(k)})}}.
\end{gather}
When the proposal distribution $\nu$ is induced by independent Bernoulli
variables $\bx_i\sim\mathrm{Bernoulli}(p_i)$, we can interpret
\eqref{eq:bin-prob} as redistributing the target bin mass $u_{l(\bz)}$
across the binary vectors whose slack values fall inside the same bin;
see \Cref{fig:slack-dist}. The remaining question is therefore how to
approximate the proposal bin masses $v_l$.

Under this independent-Bernoulli proposal, the binary-domain slack
\(s_j(\bx)=a_j^\top\bx-b_j\) has cumulant-generating function
\begin{equation}
  \label{eq:cgf}
  K_j(t)=-b_j t+\sum_{i=1}^n\log\bigl(1-p_i + p_i e^{a_{j,i} t}\bigr).
\end{equation}
For a given $s$, let \(\hat t\) solve \(K_j'(\hat t)=s\). We can then
approximate the CDF of the proposal slack distribution by the standard
Lugannani--Rice saddlepoint formula~\cite{lugannani1980},
\[
\widehat F_{\bm{s}_j}(s)\approx \Phi_{\mathrm N}(a(s))
+\phi_{\mathrm N}(a(s))\!\left(\frac{1}{a(s)}-\frac{1}{b(s)}\right),
\]
where $\Phi_{\mathrm N}$ and $\phi_{\mathrm N}$ denote the standard
normal CDF and PDF, and
\[
a(s)=\operatorname{sign}(\hat t)\sqrt{2(\hat t s - K_j(\hat t))},
\qquad
b(s)=\hat t\sqrt{K_j''(\hat t)}.
\]
We then set
\(\widehat v_l=\widehat F_{\bm{s}_j}(\ell_l^+)-\widehat F_{\bm{s}_j}(\ell_l^-)\).
Replacing \(v_l\) by \(\widehat v_l\) yields approximate weights;
because \(\widehat v_l\) may be imperfect, we use the \acrshort{snis}
estimators in \eqref{eq:est-G} and \eqref{eq:est-c}, which correct for
normalization and mitigate small approximation errors.

This construction allows us to use any target slack density
$\slackdens{j}$---including mixtures of densities---when defining the
induced spin-domain measure $\binprob{j}$.

\section{Experimental Setup}
Projecting a penalty function $\psi$ onto fewer quadratic terms necessarily yields a poorer approximation than methods that use all quadratic terms. The experiments in this section therefore ask:
\emph{does the benefit of a topology-aware penalty outweigh the downside of using fewer quadratic terms in the
  approximation?} Source files for replicating all experiments are
available at:
\url{https://github.com/lklee9/topology-aware-walsh-fourier-penalization}.

\subsection{Problem Instances}\label{sec:problems}
To evaluate and compare \acrshort{method} with existing inequality
penalization approaches, we use two families of constrained binary
optimization problems: Maximum Independent Set (MIS) and
Multidimensional Knapsack (MDKP). These problem families were chosen
because they represent distinct challenge profiles: MIS encodes pairwise
conflict constraints, while MDKP features constraints that can
potentially involve every variable. 

% For the final experiment on D-Wave's quantum annealers in
% \Cref{sec:benchmark}, we use benchmark instances that were also used in
% \cite{sharma2025}. However, to obtain smaller problem instances for our
% initial classical experiments, we use randomly generated instances of
% the two problem families designed to have properties roughly similar to
% those of the benchmark instances. Furthermore, for each problem family
% and size, we generate 20 random instances beforehand and use the same
% set of instances for all experiments that use randomly generated problem
% instances. Before using these problem instances, we convert them into
% constrained minimization problems of the form in
% \eqref{eq:opt-prob}.

For the final D-Wave hardware's quantum annealing experiments in \Cref{sec:benchmark}, we use the benchmark instances from \cite{sharma2025}. For the preceding classical experiments, we instead use smaller randomly generated instances with properties similar to the benchmark set. For each problem family and size, we generate 20 instances and use the same set throughout all experiments. Each instance is first converted into the constrained minimization form in \eqref{eq:opt-prob}.

\subsubsection{Maximum Independent Set (MIS)}
For an undirected graph \(G=(V,E)\) with \(n=|V|\), the maximum
independent set problem selects the largest subset of mutually
nonadjacent vertices:
\begin{equation}
  \label{eq:mis}
  \max_{\bx\in \{0,1\}^n} \; \sum_{i=1}^n \bx_i, 
  \qquad
  \mathrm{s.t.}\quad
  \bx_i+\bx_j \leqslant 1 \quad \forall (i,j)\in E.
\end{equation}
Here \(\bx_i=1\) indicates that vertex \(i\) is selected, and each
constraint enforces that no edge has both endpoints selected. To
generate random MIS instances, we first sample an edge density
\(\rho\sim\mathrm{Unif}[0.095,0.215]\), and then add
\(\lfloor \rho \binom{n}{2}\rfloor\) edges uniformly at random to an
undirected graph on \(n\) vertices. The density range is calibrated from
the benchmark MIS instances derived from error-correcting-code
datasets~\cite{sloane2000}.

\subsubsection{Multidimensional Knapsack Problem (MDKP)}
In the multidimensional knapsack problem, there are \(n\) items with
positive profits \(p_i>0\) and \(m\) resource constraints. Selecting
item \(i\) consumes \(w_{j,i}\geqslant 0\) units of resource \(j\), whose
capacity is \(c_j\). The problem is
\begin{equation}
  \label{eq:mdkp}
  \begin{aligned}
    \max_{\bx\in \{0,1\}^n} \; \sum_{i=1}^n p_i \bx_i,
    \:\:\mathrm{s.t.}\:\:
    \sum_{i=1}^n w_{j,i}\bx_i \leqslant c_j,
    \:\: \forall j\in[m].
  \end{aligned}
\end{equation}
Thus \(\bx_i=1\) denotes selecting item \(i\), and each constraint limits
the total consumption of one resource.

For random MDKP instances, we set \(m=n-2\). We independently sample
profits \(p_i\in\{20,21,\ldots,42000\}\) and weights
\(w_{j,i}\in\{0,1,\ldots,310\}\) uniformly. For each resource \(j\), we
sample a capacity ratio \(\rho_j\sim\mathrm{Unif}[0.4,0.8]\) and set
\(
  c_j = \left\lfloor \rho_j \sum_i w_{j,i} \right\rfloor.
\)
These parameter ranges are based on the SAC-94 MDKP benchmark suite,
which contains instances derived from real-world industrial
problems~\cite{drake2015}.

\subsection{Methods to Compare}\label{sec:methods}
Since \acrshort{method} is a slack-free penalization method for encoding
inequality constraints in QUBO problems, we will mainly compare
\acrshort{method} to a classical baseline and a state-of-the-art
slack-free penalty method. Specifically, we consider:
\begin{enumerate}[leftmargin=*, labelsep=0.5em]
\item \textbf{\acrshort{method} (full) Penalty}: Our proposed method but
  where we ignore the topology of the solver and project the origin
  penalty $\psi$ onto all quadratic terms.
\item \textbf{\acrshort{method} (topology) Penalty}: Our proposed method
  where $\psi$ is projected onto the quadratic terms implied by the
  topology of the solver.
\item \textbf{Unbalanced penalization (UP)}: A slack-free encoding of
  inequality constraints with a tunable quadratic penalty
  \cite{montanez-barrera2024},
  \begin{equation}
  \label{eq:unb-pen}
  \min_{\bx\in\{0,1\}^{n}}f(\bx) + \sum_{j=1}^{m}
  \Biggl[\frac{\hat{\lambda}_{1}}{2} \bigl(s_{j}(\bx)\bigr)^{2} -
  \hat{\lambda}_{2} \bigl(s_{j}(\bx)\bigr)\Biggr].
\end{equation}
where $s_{j}(\bx)\coloneq a_{j}^{\top}\bx - b_{j}$.
As in \cite{montanez-barrera2024}, we use the same coefficients for
every constraint, obtained via Nelder--Mead; see
\Cref{sec:qubo-build} for further details.
\item \textbf{Classical reference (IBM CPLEX v22.1.2)}: A commercial
  mixed-integer programming solver.
\end{enumerate}
We will not compare \acrshort{method} to the slack-based encoding of
inequality constraints as \cite{montanez-barrera2024} has already
compared penalty-based approaches to slack-based approaches for
encoding inequality constraints. Furthermore, we will mainly use the
classical reference method, CPLEX, to obtain the (approximate) optimal
solutions for the problem instances in \Cref{sec:problems} to help with
computing the performance metrics in \Cref{sec:metrics}.

\subsection[Implementation of TAWP Used in This Paper]{Implementation of \acrshort{method} Used in This Paper}\label{sec:method-config}
In all experiments we fix the two user inputs of \acrshort{method}
(\Cref{sec:main-res}): the origin penalty $\psi$ is the Heaviside
function, and $\mu_{\bZ}$ is induced by a Gaussian-like target
density over the slack values, centered at half of the maximum
feasible slack,
\begin{equation}
  \label{eq:gaussian-target}
  \slackdens{j}(s) \propto
  \exp\Bigl(-\tfrac{(s-c)^2}{2\sigma^2}\Bigr),
  \qquad c = \tfrac{1}{2}s_j^{\max},
\end{equation}
where \(s = a_j^\top(1+\bz)/2 - b_j\),
\(s_j^{\max}=\sum_{i:a_{j,i}>0}a_{j,i}-b_j\) is the maximum attainable
slack, and the per-constraint bandwidth $\sigma > 0$ is chosen such that
at least \(95\%\) of the density mass falls on the feasible range. From
this slack density we obtain $\mu_{\bZ}$ using the Lugannani--Rice
approximation described in \Cref{subsec:target-measure}.  The hope is
that by concentrating the error budget in the middle of the feasible
range, the measure might encourage the projected quadratic surrogate to
be a roughly symmetric bowl over the feasible region and to grow large on
the infeasible side---where approximation accuracy is irrelevant as long
as the penalty is large enough.

\begin{figure}
  \centering

  \subfloat[Original constraint: the feasible side dominates the slack
  range.]{%
    \includegraphics[width=0.9\linewidth]{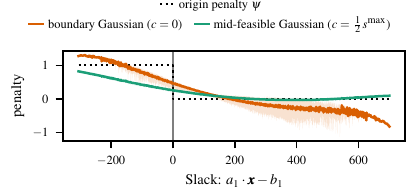}%
    \label{fig:measure-projection:a}}

  \subfloat[Modified constraint: the infeasible side dominates the slack
  range.]{%
    \includegraphics[width=0.9\linewidth]{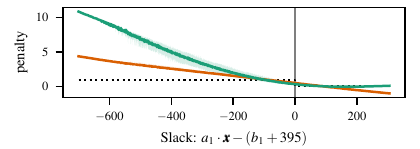}%
    \label{fig:measure-projection:b}}

  \caption{Exact projection of Heaviside penalty
    \(\psi\) onto full quadratic subspace for constraint \(j=1\) of
    \texttt{pet3}. Lines: mean projected penalty per slack value; bands:
    min--max range.}
  \label{fig:measure-projection}
\end{figure}

\Cref{fig:measure-projection} compares exact-enumeration projections of
the same Heaviside penalty under two Gaussian target measures for one
illustrative \texttt{pet3} constraint: once with the feasible side
dominating the slack range, and once with the constraint tightened so
that the infeasible side dominates. In both situations, the
boundary-centered Gaussian mis-shapes the feasible landscape and rewards
increasing amounts of slack---distorting the ranking of feasible
solutions. The mid-feasible Gaussian in \eqref{eq:gaussian-target}
instead realizes the intended penalty in both cases, i.e. it is
approximately flat over the feasible range, and gradually increasing on
violations. That said, we emphasize that the target in
\eqref{eq:gaussian-target} is one simple choice rather than a uniquely
principled one.

The proposal distribution used for \acrshort{snis} is the product
Bernoulli law with \(p_i=0.5\). The moments in
\eqref{eq:est-G} and \eqref{eq:est-c} are estimated from \(N=\min(2^{n-1},2^{15})\)
proposal samples, and the normal equations are solved with a
ridge parameter \(\epsilon=10^{-8}\). For topology-aware variants of
\acrshort{method}, the coupling set $\edges$ is induced from the
target hardware graph by a simple greedy logical-to-physical mapping
(\Cref{apx:mapping}), providing a baseline with minimal effort devoted to the
mapping and the choices of $\psi$ and $\mu_{\bZ}$.

These choices fix the \acrshort{method} configuration throughout the
paper, leaving only the outer inequality multiplier to be tuned, as
described in the following section.

\subsection{QUBO Construction}\label{sec:qubo-build}
\begin{table}
  \centering
  \caption{Tuned global penalty multipliers per problem family and
    method; $\hat{\lambda}_{1},\hat{\lambda}_{2}$ are the additional
    shape parameters of unbalanced penalization.}
  \label{tab:mults}
  \setlength{\tabcolsep}{3pt}
  \footnotesize
  \begin{tabular}{@{}lccccccc@{}}
    \toprule
    \textbf{Family} & \multicolumn{4}{c}{\textbf{\acrshort{method}}} & \multicolumn{3}{c}{\textbf{Unbalanced P.}}\\
    \cmidrule(lr){2-5}\cmidrule(lr){6-8}
    & \multicolumn{1}{c}{\textit{Full}} & \multicolumn{1}{c}{\textit{Chim.}} & \multicolumn{1}{c}{\textit{Peg.}} & \multicolumn{1}{c}{\textit{Zeph.}} & \multicolumn{1}{c}{$\lambda$} & \multicolumn{1}{c}{$\hat{\lambda}_{1}$} & \multicolumn{1}{c}{$\hat{\lambda}_{2}$}\\
    \midrule
    MDKP & 0.1306 & 0.0621 & 0.0632 & 0.0660 & 0.0998 & 0.9997 & 0.0003\\
    MIS  & 0.2167 & 0.2163 & 0.2167 & 0.2165 & 0.1920 & 0.6250 & 0.3750\\
    \bottomrule
  \end{tabular}
\end{table}

Using either unbalanced penalization or \acrshort{method}, we obtain the following QUBO matrix from each problem instance:
\begin{equation}
  \label{eq:qubo-sum}
  Q + \lambda\sum_{j=1}^{m}\frac{\sigma_{Q}}{\sigma_{j}}A_{j}.
\end{equation}
Here, $Q$ is the QUBO matrix for the objective and $A_{j}$ are
the penalty matrices for the $m$ inequality constraints, obtained
via unbalanced penalization or \acrshort{method}. Before summing the penalty matrices, we normalize each $A_{j}$ by its standard deviation $\sigma_{j}$ under a uniform
distribution over $\{0,1\}^{n}$~\cite{lee2025b}. This helps to ensure
that each constraint has roughly the same scale, particularly
helpful for MDKP-like problems. We further rescale each normalized constraint penalty by the standard deviation of the objective function, $\sigma_Q$, so that the penalties and objective are on comparable scales.

The penalty multipliers in \Cref{tab:mults} were obtained by tuning on a
fixed problem size for each problem family, specifically MDKP ($n=15$)
and MIS ($n=16$). For unbalanced
penalization, we also tuned the parameters in \eqref{eq:unb-pen} that
control the shape of its quadratic penalty. These multipliers were found
using Nelder--Mead to minimize the optimality gap defined in
\Cref{sec:metrics}.

\subsection{Metrics}\label{sec:metrics}
To evaluate the solution quality and consistency of each method in
\Cref{sec:methods}, we use the following metrics:
\begin{enumerate}
\item \textbf{Optimality Gap (Per Instance)}~\cite{
    sharma2025,lee2025a,takabayashi2025a,takabayashi2025}:
  \begin{equation}
    \label{eq:5}
    \frac{\text{Best objective found} - \text{True optimum objective}}{\text{True optimum objective}}.
  \end{equation}
% \item \textbf{Number of Physical Qubits After Embedding}~\cite{suda2026}
\item \textbf{\acrfull{cop}}:
  Used by \cite{montanez-barrera2024} for assessing their method of
  unbalanced penalization, \acrshort{cop} is basically the optimality
  rate scaled by the size of the solution space,
  \begin{equation}
    \label{eq:10}
    2^{n} \times \frac{\text{\# reads with true optimal solution}}
    {\text{\# reads}}.
  \end{equation}
  \acrshort{cop} also converts monotonically to the standard
  \acrfull{tts}~\cite{ronnow2014}. Letting
  \(\bar{p}=\text{mean }\acrshort{cop}/2^{n}\) pooled over all problem
  instances, observing the optimum with \(99\%\) probability takes
  \(R_{99}=\ln(0.01)/\ln(1-\bar{p})\) reads, and
  \(\mathrm{TTS}_{99}=t_{\mathrm{a}}R_{99}\) time for annealing time
  \(t_{\mathrm{a}}\).
% \item \textbf{Feasibility Rate}~\cite{lee2025a,suda2026}
%   \begin{equation}
%     \label{eq:fea-rate}
%     \frac{\text{\# Instances with feasible solution}}
%     {\text{\# Instances}}
%   \end{equation}
\end{enumerate}
Furthermore, when comparing different methods against a ``base''
method, we use the following $\Delta$ formulas to compute the relative
change of a metric between the methods,
\begin{gather}
  \label{eq:delta-gap}
  \Delta\mathrm{gap} = (\text{Base Gap} - \text{Method Gap})/\text{Base Gap},\\
  \label{eq:delta-cop}
  \Delta\mathrm{CoP} = (\text{Method CoP} - \text{Base CoP})/\text{Base CoP}.
\end{gather}
A positive relative change favors the method being compared against the
base method.

\section{Experiments}\label{sec:experiments}
\begin{figure*}[t]
  \centering
  
  \subfloat[Relative difference in \acrshort{cop} between
  \acrshort{methodfull} and unbalanced penalization when solved using
  \acrshort{sqa} before and after minor embedding on different
  hardware topologies.]{%
    \includegraphics[height=0.17\paperheight]{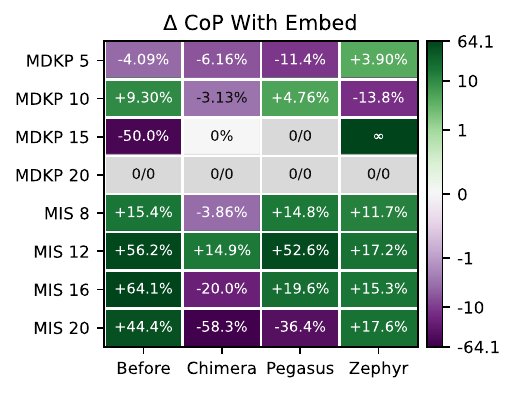}
    \label{fig:compare-up:a}}
  \hspace{1em}
  \subfloat[Relative difference in \acrshort{cop} between
  \acrshort{methodtopo} and unbalanced penalization when solved using
  \acrshort{sqa} before and after being embedded on different hardware
  topologies.]{%
    \includegraphics[height=0.17\paperheight]{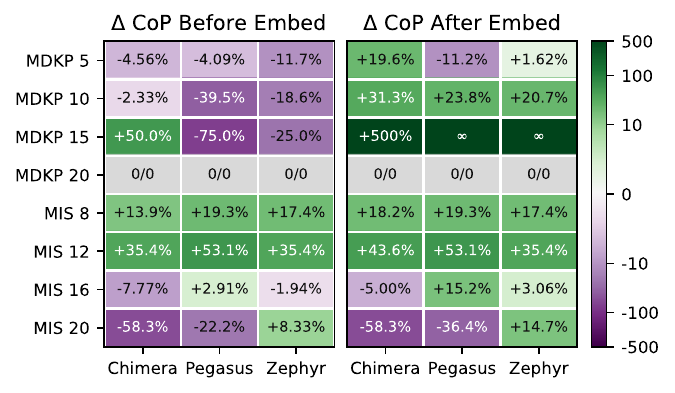}
    \label{fig:compare-up:b}}
  \caption{Mean relative \acrfull{cop} difference against the
    unbalanced-penalization base; negative values favor unbalanced
    penalization. Cells labeled $0/0$ have zero mean \acrshort{cop} for
    both methods.}
  \label{fig:compare-up}
\end{figure*}

This section addresses the following two questions about our proposed
\acrfull{method} framework:
\begin{enumerate}
\item How does the quadratic penalty formed by projecting the Heaviside
  function compare with unbalanced penalization in terms of solution
  quality?
\item When embedding a QUBO on a given topology, does the benefit
  of \acrshort{method} penalty outweigh the loss of using
  fewer quadratic terms in the approximation?
\end{enumerate}

Throughout our experiments, we use D-Wave's path-integral \acrfull{sqa}
with a default chain strength computed by
\texttt{uniform\_torque\_compensation} to sample solutions to the
constructed QUBO problems with 1000 reads per instance. This lets us
\emph{partially} simulate the behavior of a quantum annealer---including
chain breaks over different topologies---on classical hardware. The
topologies we will use are: (1) a fully connected topology serving as a
baseline with no embedding overhead, and (2) the D-Wave hardware
topologies Chimera, Pegasus, and Zephyr as described in
\Cref{sec:back-qa}. When using these hardware topologies, we use
\texttt{minorminer} to embed the problem onto the target topology before
sampling solutions with \acrshort{sqa}. However, when using
\acrshort{method} penalties that are projected onto the topology itself,
the embedding step is direct and therefore does not produce any chains
or additional qubit overhead.

For the remainder of this section, we present the main results; the full results are available in the code repository.

\subsection[Penalty Quality: TAWP vs. Unbalanced Penalization]{Penalty Quality: \acrshort{method} vs. Unbalanced Penalization}
% \begin{figure}
%   \centering
%   \includegraphics[width=\linewidth]{figures/full_vs_unbalanced_gap_boxplots_mdkp.pdf}
%   \caption{Distribution of optimality gap for MDKP problem instances comparing TAWP (projected Heaviside) against unbalanced penalization (UP). Lower gaps indicate better solution quality.}
%   \label{fig:unb-pen-strip}
% \end{figure}

One of our main empirical questions is how \acrfull{method}, used to
construct a quadratic approximation of an inequality penalty function
\(\psi\), compares with unbalanced penalization in terms of solution
quality. Following the implementation described in
\Cref{sec:method-config}, we compare unbalanced penalization against
the \acrshort{method} penalty obtained by projecting a Heaviside
penalty under a Gaussian measure centered at the midpoint of the
feasible slack region.

To address this question, \Cref{fig:compare-up} reports the relative
difference in \acrshort{cop} between the \acrshort{methodfull} and
\acrfull{up} penalties, with positive values indicating an advantage for
\acrshort{methodfull}. As shown in \Cref{fig:compare-up:a}, before
embedding, the \acrshort{methodfull} penalty outperforms \acrshort{up} on
average across all MIS problem sizes, but underperforms on the MDKP
instances. This suggests that, when projected onto the full quadratic
space, \acrshort{methodfull} is less effective for the MDKP
inequality-constraint penalties. Moreover, its performance degrades on
sparser target topologies such as Chimera, while remaining competitive
on denser topologies such as Zephyr. This behavior suggests that
\acrshort{methodfull} penalties can be sensitive to embedding-induced
errors when the target hardware topology is sparser than the logical
interaction graph induced by the penalty.

This issue can be mitigated by projecting the Heaviside penalty directly
onto the topology on which the QUBO will be embedded, rather than first
projecting onto the full set of quadratic terms. As shown in
\Cref{fig:compare-up:b}, this topology-aware projection preserves the
advantage of \acrshort{method} on the MIS instances and generally
improves its performance on the MDKP instances as well. In particular,
on the Zephyr topology, \acrshort{methodtopo} outperforms \acrshort{up}
on every problem instance except MDKP(20), where the two methods remain
close, with mean objective gaps of \(0.1199\) and \(0.1245\) for
\acrshort{methodtopo} and \acrshort{up}, respectively.

\subsection{Benefits of Topology-Aware Penalties}
\begin{figure}
  \centering
  \includegraphics[width=0.9\linewidth]{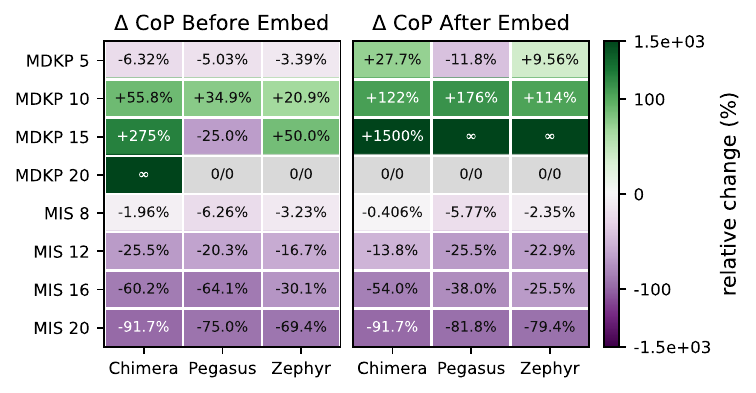}
  \caption{Relative \acrshort{cop} difference before and after
    embedding on each D-Wave topology, between \acrshort{up} and
    \acrshort{methodtopo} with \acrshort{up}'s quadratic penalty as
    the origin; positive favors \acrshort{methodtopo}.}
  \label{fig:obj-gap}
\end{figure}

\begin{table}[t]
  \centering
  \caption{Mean \acrshort{cop} before/after embedding on Zephyr for
    \acrshort{up} and its \acrshort{methodtopo} projection; higher is
    better.}
  \label{tab:compare-unb-pen-cop-mean}
  \begin{tabular}{@{}llrrrr@{}}
  \toprule
  & &  \multicolumn{2}{c}{Logical} & \multicolumn{2}{c}{Embedding} \\
  \cmidrule(lr){3-4} \cmidrule(lr){5-6}
  Family & $n$ & \acrshort{up} & \acrshort{method} & \acrshort{up} & \acrshort{method} \\
  \midrule
  MDKP & 5 & 1.37 & 1.32 & 1.19 & 1.3 \\
   % &  & Fea. & 0.27 & 0.27 & 0.22 & 0.26 \\
  MDKP & 10 & 2.2 & 2.66 & 1.48 & 3.17 \\
   % &  & Fea. & 0.18 & 0.18 & 0.05 & 0.18 \\
  MDKP & 15 & 6.55 & 9.83 & 0 & 18.02 \\
   % &  & Fea. & 0.23 & 0.26 & 0.05 & 0.26 \\
  MDKP & 20 & 0 & 0 & 0 & 0 \\
   % &  & Fea. & 0.25 & 0.32 & 0.1 & 0.32 \\
  \midrule
  MIS & 8 & 13.08 & 12.66 & 13.08 & 12.77 \\
   % &  & Fea. & 0.3 & 0.3 & 0.3 & 0.31 \\
  MIS & 12 & 39.32 & 32.77 & 39.32 & 30.31 \\
   % &  & Fea. & 0.17 & 0.18 & 0.17 & 0.17 \\
  MIS & 16 & 337.51 & 235.93 & 321.13 & 239.21 \\
   % &  & Fea. & 0.13 & 0.14 & 0.13 & 0.14 \\
  MIS & 20 & 1887.44 & 576.72 & 1782.58 & 367 \\
   % &  & Fea. & 0.14 & 0.11 & 0.13 & 0.1 \\
  \bottomrule
\end{tabular}

\end{table}

In the previous section, we compared unbalanced penalization with a \acrshort{method} penalty obtained by projecting the Heaviside function onto all quadratic terms. However, this comparison combines two potential benefits of the \acrshort{method} framework: (1) projecting the Heaviside function, and (2) restricting the projection to the quadratic terms supported by the target topology. To isolate the effect of the topology restriction alone, we replace the Heaviside function with the quadratic penalty from unbalanced penalization and project this penalty onto the target topologies.

% \Cref{fig:obj-gap} shows the mean relative change in \acrshort{cop}
% between the \acrshort{methodtopo} penalty and unbalanced
% penalization. Surprisingly, we observe the opposite pattern from the
% previous section: \acrshort{methodtopo} now generally performs better on
% MDKP instances but worse on MIS instances. This reversal implies that
% the improvements we observed for MIS in the previous section were primarily
% due to the Heaviside projection itself, rather than the topology-aware
% restriction. Conversely, the improvements on MDKP with topology-aware
% penalties, even before embedding the resulting QUBOs on a specific
% topology, can likely be attributed to the regularization effect of
% dropping low-importance quadratic terms during the projection.

\Cref{fig:obj-gap} shows the mean relative change in \acrshort{cop} between \acrshort{methodtopo} and unbalanced penalization. Surprisingly, the trend reverses from the previous section: \acrshort{methodtopo} generally performs better on MDKP but worse on MIS. This suggests that the earlier MIS improvements were primarily due to the Heaviside projection rather than the topology-aware restriction. In contrast, the MDKP improvements likely stem from the regularization effect of discarding low-importance quadratic terms during the projection, even before embedding the resulting QUBOs onto a specific topology.

% Part of the reason \acrshort{methodtopo} underperforms on the MIS
% instances when compared to \acrshort{up} is the inherent structure of
% the constraints in each problem family. In MIS, each inequality involves
% only two variables, making the constraints relatively simple but also
% making them sensitive to any approximation error. In MDKP, each
% inequality can involve many variables, making the constraints inherently
% ``harder'' and the penalty landscape more complex. We can empirically
% observe the consequences of this ``harder'' penalty landscape in the
% smaller absolute \acrshort{cop} values in
% \Cref{tab:compare-unb-pen-cop-mean}. Therefore, \acrshort{methodtopo}
% seems to perform better on ``harder'' inequality constraints while still
% performing acceptably on ``easier'' ones---it still is capable of
% finding the true feasible optimum for the MIS problems.

One reason \acrshort{methodtopo} underperforms \acrshort{up} on MIS is the structure of the constraints. MIS constraints involve only two variables and are therefore more sensitive to approximation error. In contrast, MDKP constraints couple many variables, creating a more complex penalty landscape, as reflected by the smaller absolute \acrshort{cop} values in \Cref{tab:compare-unb-pen-cop-mean}. These results suggest that \acrshort{methodtopo} is most beneficial for more dense inequality constraints while remaining competitive on sparser ones. Even for MIS, it still recovers the true feasible optimum.

% In \acrshort{tts} terms (\Cref{sec:metrics}), the embedding cost is
% stark: at MDKP \(n=15\), embedded \acrshort{up} finds no optimum
% within \(2\times10^{4}\) pooled reads, whereas native
% \acrshort{methodtopo} reaches the optimum within
% \(R_{99}\approx8\times10^{3}\) reads.

\subsection{Experiment: Benchmark on D-Wave Hardware}\label{sec:benchmark}

\begin{table*}[t]
\centering
\caption{Objective gaps (\%) for MDKP benchmarks on D-Wave Advantage
  and Advantage2. Cells: sample-level \(\mathrm{mean}[\min,\max]\)
  over \(2500\) reads (lower is better), with the feasible-sample
  count on a second line. Bold marks the best mean/min/max gap and
  the largest feasible count per instance and QPU; ``--'' means no
  feasible sample.}
\label{tab:mdkp-gaps}
\scriptsize
\setlength{\tabcolsep}{3pt}
\newcommand{\gapcell}[2]{\begin{tabular}[t]{@{}r@{}}#1\\{\scriptsize \(#2/2500\)}\end{tabular}}
\begin{tabular}{@{}llrrrrrr@{}}
  \toprule
  & & \multicolumn{3}{c}{Advantage (Pegasus)} & \multicolumn{3}{c}{Advantage2 (Zephyr)} \\
  \cmidrule(lr){3-5} \cmidrule(lr){6-8}
  Instance & \(n\) & \acrshort{up} \cite{montanez-barrera2024} & \acrshort{methodfull} & \acrshort{methodtopo} &  \acrshort{up} \cite{montanez-barrera2024} & \acrshort{methodfull} & \acrshort{methodtopo} \\
  \midrule
  \begin{tabular}[t]{@{}l@{}}pet3\end{tabular} & \begin{tabular}[t]{@{}l@{}}15\end{tabular}
  & \gapcell{\textbf{14.04} [0.25, \textbf{62.02}]}{1657}
  & \gapcell{29.48 [0.25, 80.82]}{\textbf{1695}}
  & \gapcell{20.66 [\textbf{0.00}, 72.85]}{1522}
  & \gapcell{\textbf{15.81} [\textbf{0.00}, \textbf{42.71}]}{\textbf{2033}}
  & \gapcell{25.80 [0.25, 73.60]}{1745}
  & \gapcell{17.76 [\textbf{0.00}, 56.91]}{1352} \\

  \begin{tabular}[t]{@{}l@{}}pb5\end{tabular} & \begin{tabular}[t]{@{}l@{}}20\end{tabular}
  & \gapcell{16.67 [\textbf{3.74}, 33.05]}{21}
  & \gapcell{23.66 [\textbf{3.74}, 45.77]}{\textbf{282}}
  & \gapcell{\textbf{9.35} [9.35, \textbf{9.35}]}{1}
  & \gapcell{15.66 [13.32, 18.00]}{2}
  & \gapcell{19.57 [\textbf{2.52}, 44.13]}{\textbf{216}}
  & \gapcell{\textbf{14.40} [13.18, \textbf{15.61}]}{2} \\

  \begin{tabular}[t]{@{}l@{}}pet4\end{tabular} & \begin{tabular}[t]{@{}l@{}}20\end{tabular}
  & \gapcell{24.02 [1.31, 86.11]}{\textbf{1021}}
  & \gapcell{32.48 [\textbf{0.98}, 77.45]}{709}
  & \gapcell{\textbf{23.51} [1.47, \textbf{64.79}]}{397}
  & \gapcell{21.87 [\textbf{1.96}, 63.89]}{742}
  & \gapcell{26.19 [3.27, 87.42]}{\textbf{759}}
  & \gapcell{\textbf{18.04} [2.61, \textbf{56.13}]}{293} \\

  \begin{tabular}[t]{@{}l@{}}pb1\end{tabular} & \begin{tabular}[t]{@{}l@{}}27\end{tabular}
  & \gapcell{19.77 [7.51, 49.16]}{36}
  & \gapcell{22.65 [6.76, 56.60]}{\textbf{319}}
  & \gapcell{\textbf{19.18} [\textbf{5.63}, \textbf{39.19}]}{93}
  & \gapcell{18.58 [7.86, 37.44]}{21}
  & \gapcell{23.86 [\textbf{6.18}, 66.41]}{\textbf{489}}
  & \gapcell{\textbf{15.95} [8.16, \textbf{30.39}]}{21} \\

  \begin{tabular}[t]{@{}l@{}}hp1\end{tabular} & \begin{tabular}[t]{@{}l@{}}28\end{tabular}
  & \gapcell{24.24 [5.88, 54.33]}{93}
  & \gapcell{23.25 [\textbf{3.86}, 63.90]}{\textbf{620}}
  & \gapcell{\textbf{19.35} [5.38, \textbf{49.18}]}{117}
  & \gapcell{19.89 [5.82, 35.25]}{17}
  & \gapcell{21.18 [\textbf{3.39}, 59.60]}{\textbf{344}}
  & \gapcell{\textbf{16.73} [6.82, \textbf{35.17}]}{31} \\

  \begin{tabular}[t]{@{}l@{}}pet5\end{tabular} & \begin{tabular}[t]{@{}l@{}}28\end{tabular}
  & \gapcell{24.09 [\textbf{1.94}, 65.08]}{\textbf{1520}}
  & \gapcell{34.49 [7.58, 77.26]}{1472}
  & \gapcell{\textbf{22.69} [2.10, \textbf{59.64}]}{1439}
  & \gapcell{23.31 [3.51, \textbf{62.70}]}{1972}
  & \gapcell{35.04 [5.73, 83.06]}{\textbf{2061}}
  & \gapcell{\textbf{19.72} [\textbf{3.43}, 65.08]}{1521} \\

  \begin{tabular}[t]{@{}l@{}}pb2\end{tabular} & \begin{tabular}[t]{@{}l@{}}34\end{tabular}
  & --
  & \gapcell{24.89 [\textbf{10.70}, 51.60]}{\textbf{138}}
  & \gapcell{\textbf{15.29} [15.29, \textbf{15.29}]}{1}
  & --
  & \gapcell{\textbf{21.06} [\textbf{7.69}, \textbf{37.70}]}{\textbf{27}}
  & -- \\

  \begin{tabular}[t]{@{}l@{}}pet6\end{tabular} & \begin{tabular}[t]{@{}l@{}}39\end{tabular}
  & \gapcell{24.72 [\textbf{4.99}, 70.60]}{484}
  & \gapcell{27.25 [5.28, 78.81]}{\textbf{1007}}
  & \gapcell{\textbf{20.86} [5.97, \textbf{59.50}]}{469}
  & \gapcell{\textbf{16.31} [2.86, \textbf{45.53}]}{165}
  & \gapcell{28.28 [\textbf{2.19}, 79.63]}{\textbf{814}}
  & \gapcell{17.82 [4.42, 52.93]}{432} \\

  \begin{tabular}[t]{@{}l@{}}pet7\end{tabular} & \begin{tabular}[t]{@{}l@{}}50\end{tabular}
  & \gapcell{27.23 [6.95, 67.08]}{342}
  & \gapcell{31.01 [6.57, 69.83]}{\textbf{896}}
  & \gapcell{\textbf{20.06} [\textbf{5.49}, \textbf{53.64}]}{305}
  & \gapcell{22.12 [\textbf{3.94}, 62.79]}{358}
  & \gapcell{28.36 [4.34, 64.62]}{\textbf{817}}
  & \gapcell{\textbf{15.94} [5.91, \textbf{48.02}]}{188} \\
  \bottomrule
\end{tabular}
\end{table*}

% To investigate this
% phenomenon further, we next evaluated the three penalization methods
% from \Cref{sec:methods}---\acrshort{up}, \acrshort{methodfull}, and
% \acrshort{methodtopo}---on the benchmark MDKP instances described in
% \Cref{sec:problems}. For the Pegasus hardware family we used D-Wave's
% Advantage system, and for the Zephyr hardware family we used D-Wave's
% Advantage2 system. In both cases, each QUBO was submitted with an
% annealing time of \(25\,\mu\mathrm{s}\) and \(2500\) reads.

% The sample-level objective-gap statistics are reported in
% \Cref{tab:mdkp-gaps}. A clear pattern emerges: \acrshort{methodtopo}
% achieves the lowest mean objective gap for almost every available
% instance--hardware pair. The main exceptions are instance pet3 on both
% systems; and instances pet6/pb2 on Advantage2. Moreover,
% \acrshort{methodtopo} also achieves the smallest \emph{maximum}
% objective gap for most instances, which indicates that it is typically
% the ``least risky'' method in the sense of avoiding very poor returned
% samples.

% At the same time, we can see that \acrshort{methodtopo} is usually not
% the method that finds the single best sample nor the one that returns
% the most feasible samples. In many cases, \acrshort{methodfull} is the
% one that dominates both those metrics, even when it has a worse mean
% and worst-case behavior overall. Thus, even for metrics that
% \acrshort{methodtopo} does not perform well in, \acrshort{methodfull}
% still dominates \acrshort{up} in those metrics most of the time.

To investigate this phenomenon further, we next evaluated the three
penalization methods from
\Cref{sec:methods}---\acrshort{up}, \acrshort{methodfull}, and
\acrshort{methodtopo}---on the benchmark MDKP instances described in
\Cref{sec:problems}. For the Pegasus hardware family we used
D-Wave's Advantage system, and for the Zephyr hardware family we
used D-Wave's Advantage2 system. In both cases, each QUBO was
submitted with an annealing time of \(25\,\mu\mathrm{s}\) and
\(2500\) reads.

The sample-level objective-gap statistics are reported in
\Cref{tab:mdkp-gaps}. A clear pattern emerges: \acrshort{methodtopo}
attains the lowest mean objective gap on \(8\) of \(9\) instances on
Advantage and \(6\) of \(8\) comparable instances on Advantage2, and
also improves the worst observed gap on most instances, indicating
that it is typically the ``least risky'' method in the sense of
avoiding very poor returned samples. 

The feasibility counts show the complementary side of this trade-off.
\acrshort{methodtopo} is usually neither the method that finds the
single best sample nor the one that returns the most feasible samples;
\acrshort{methodfull} often dominates these metrics, even when its mean
and worst observed gaps are worse. This pattern is consistent with the
projection structure. Since \acrshort{methodfull} projects onto the full
pairwise space, it is the closest quadratic surrogate to \(\psi\) in the
chosen weighted least-squares sense, which can improve feasible-sample
yield. In contrast, \acrshort{methodtopo} restricts the admissible
quadratic terms to the hardware edge set \(\edges\), reducing the
approximation space and increasing the residual
\(\|\psi-\hat\psi_{\edges}\|_\mu^2\). This can shift the energy ordering
of near-boundary configurations and reduce the probability of sampling
feasible points. However, \acrshort{methodtopo} is implemented directly
on native couplers, avoiding chain-break postprocessing and the
dynamic-range trade-offs from chain-strength tuning that can degrade
\acrshort{methodfull} after embedding. Thus, on the MDKP workloads,
\acrshort{methodtopo} trades feasible-sample yield for improved objective
quality among returned feasible samples, so we report feasibility rate
and feasible-solution quality as complementary metrics.

\section{Conclusions}

We presented \acrshort{method}, a topology-aware Walsh--Fourier penalization method for constructing slack-free, hardware-aware QUBO penalties for constrained binary optimization. Expressing a user-chosen penalty in the Walsh--Fourier basis and projecting it onto the linear and quadratic terms supported by a target topology yields a uniquely defined weighted least-squares surrogate, removing both auxiliary slack variables and per-constraint penalty tuning. Although evaluated here on inequality constraints, the construction applies to arbitrary pseudo-Boolean penalties, including equality constraints and, in principle, objectives themselves.

Empirically, topology awareness pays off once embedding is taken into account. Full-quadratic projections are already competitive with unbalanced penalization, while topology-aware projections often deliver better overall solution quality, with the clearest gains on denser constraint families such as MDKP. The hardware results on D-Wave Advantage and Advantage2 reinforce this trend: \acrshort{methodtopo} attains the lowest mean objective gap on most benchmark instances and usually the smallest worst-case gap, indicating reliable overall performance, whereas \acrshort{methodfull} more often produces the best sample and the largest number of feasible solutions.

\acrshort{method} eliminates per-constraint penalty tuning but not design freedom: the choices of \(\psi\) and \(\mu_{\bZ}\), the greedy placement, and the per-family multiplier remain heuristic (\Cref{sec:method-config,sec:qubo-build}). The optimality guarantee of \Cref{subsec:exact-projection} is conditional on them; the sampled implementation of \Cref{alg:wls} introduces Monte Carlo and regularization error that we do not bound; and a small weighted least-squares error does not by itself certify correct feasible/infeasible ordering. Feasibility preservation and solution quality are therefore validated only empirically. This motivates \textit{several directions for future work}: conditions on \((\psi,\mu_{\bZ},\edges)\) that provably separate feasible from infeasible assignments with a prescribed margin, together with a characterization of how such guarantees degrade as \(\edges\) becomes sparser; adaptive or iterative projection schemes; and intermediate feature sets that trade approximation fidelity against hardware efficiency.

In summary, our findings suggest that a QUBO reformulation for near-term quantum annealers should be judged not only by its fit to an ideal penalty in a fully connected space, but by how well it matches the connectivity of the target device. The Walsh--Fourier projection framework makes this trade-off explicit and provides a principled rigorous route toward more deployable constraint encodings.

\appendices
\crefalias{section}{appendix}

\section{\texorpdfstring{Complexity of \acrshort{method}}{Complexity of TAWP}}\label{apx:proofs}

\begin{proposition}[Complexity of sampled Walsh-coordinate
  WLS]\label{prop:complexity}
  Let $\edges$ be the set of admissible logical couplings, let
  $\Phi\coloneq \{\emptyset\}\cup \bigl\{\{i\}:i\in[n]\bigr\}\cup \edges$
  denote the admissible Walsh features, and let $K=|\Phi|$. For the
  implementation used in this paper, which materializes the design
  matrix $\mathbf{X}_{\Phi}\in\mathbb{R}^{N\times K}$ whose $k$th row is
  $\bm{g}_{\Phi}(\bz^{(k)})^{\top}$, the sampled weighted
  least-squares projection of $\psi$ can be computed in
  $\mathcal{O}(NK^2 + K^3)$ time and $\mathcal{O}(NK + K^2)$ memory.
\end{proposition}
\begin{proof}
  Building and storing $\mathbf{X}_{\Phi}$ costs $\mathcal{O}(NK)$ time and
  memory. Forming the Gram matrix
  $\hat G=\mathbf{X}_{\Phi}^{\top}W\mathbf{X}_{\Phi}$ costs
  $\mathcal{O}(NK^{2})$ time and $\mathcal{O}(K^{2})$ storage, where $W$ is the diagonal
  matrix of normalized sample weights. The correlation vector
  $\bm{\hat c}=\mathbf{X}_{\Phi}^{\top}W\bm{p}$, with
  $p_k=\psi(\bz^{(k)})$, costs only $\mathcal{O}(NK)$ additional time and is
  therefore dominated by the Gram-matrix step. Solving the regularized
  normal equations
  $(\hat G+\epsilon I)\bm{\hat{\theta}}=\bm{\hat c}$ for $K$ unknowns has
  worst-case cost $\mathcal{O}(K^3)$. Hence total complexity is
  $\mathcal{O}(NK^2 + K^3)$ time and $\mathcal{O}(NK + K^2)$ memory.
\end{proof}

\section{Logical-to-Physical Mapping}
\label{apx:mapping}
For topology-aware projection we construct an injective placement
$\pi:V_{\mathrm L}\hookrightarrow V_{\mathrm H}$ with a greedy
multi-start heuristic. We first build a weighted logical graph whose
edge weights accumulate $|A_{ti}|\,|A_{tj}|$ over all constraint rows
$t$ whose support contains $\{i,j\}$, so that variables that
repeatedly co-occur with large coefficients are drawn onto adjacent
hardware vertices. Starting from several high-scoring seed pairs, the
placement is extended greedily---always assigning the unplaced
variable with the largest total weight to already placed variables
onto the available hardware vertex that best preserves direct
adjacency weight---and then refined by a short pairwise-swap local
search that maximizes the total preserved coupling weight. The best
placement across restarts induces the logical coupling set
$\edges=\{\{i,j\}\subseteq V_{\mathrm L}:
(\pi(i),\pi(j))\in\edges_{\mathrm H}\}$ used in the main text, so
the projected QUBO retains only quadratic terms realizable on the
chosen placement and requires no chains. The exact scoring functions,
tie-breaking rules, and parameter values are documented in the
accompanying implementation
(\texttt{fourier\_projection/greedy\_mapping.py} in the source
repository referenced in \Cref{sec:problems}).

\section*{Acknowledgements}
This research was funded by the Federal Ministry of Research, Technology and Space of Germany and the state of North Rhine-Westphalia as part of the Lamarr Institute for Machine Learning and Artificial Intelligence. The authors also acknowledge support from the United States Department of Energy's LANL Laboratory Directed Research and Development (LDRD) program through Project 20240032DR, ``Accelerating Scientific Discovery with Quantum Annealing." The authors thank Max Bannach for assistance with the D-Wave quantum annealing experiments.
% This research has been funded by the Federal Ministry of Research,
% Technology and Space of Germany and the state of North Rhine-Westphalia
% as part of the Lamarr Institute for Machine Learning and Artificial
% Intelligence. The authors also gratefully acknowledge support from the
% U.S. Department of Energy's LANL Laboratory Directed Research and
% Development (LDRD) program through Project 20240032DR, ``Accelerating
% Scientific Discovery with Quantum Annealing''. The authors thank Max
% Bannach for support in running the experiments on a D-Wave quantum
% annealing system.

% \begin{algorithm}[t]
% \caption{Greedy logical-to-physical mapping}
% \label{alg:greedy-mapping}
% \KwIn{Constraint matrix $M$; hardware graph $G_H=(V_H,E_H)$}
% \KwOut{Injective placement $\phi$ and induced topology $E_{\mathrm{proj}}$}
% Build $G_L$ with weights $w_{ij}$ from constraint co-occurrence in $M$\;
% Keep the largest connected component of $G_H$\;
% Score logical vertices by $s_i$ and hardware vertices by $c_q$\;
% \For{each seed pair $(i_0,q_0)$ among the top-ranked candidates}{
%   initialize $\phi(i_0)=q_0$\;
%   greedily map remaining variables from the hardware frontier to
%   maximize preserved weighted adjacency\;
%   apply a small improving pairwise-swap local search\;
%   store the best-scoring placement under $F(\phi)$\;
% }
% Return the best placement and
% $E_{\mathrm{proj}}=\{(i,j):(\phi(i),\phi(j))\in E_H\}$\;
% \end{algorithm}

\IEEEtriggeratref{20}
\bibliographystyle{IEEEtran}
\bibliography{lib, references}

\end{document}